\documentclass[reqno]{amsart}

\usepackage{a4wide}
\usepackage{amsmath,amssymb,amsfonts,amsthm}
\usepackage{moreverb,rotating,graphics}
\usepackage[utf8]{inputenc}
\usepackage[english]{babel} 
\usepackage{framed,fancybox}
\usepackage{caption, subcaption}
\usepackage{mathrsfs, esint, dsfont}
\usepackage{eurosym}
\usepackage[colorlinks, linkcolor={blue}, citecolor={red}, urlcolor = {blue}]{hyperref}
\usepackage{cancel}
\usepackage{bbm}
\usepackage{caption}
\usepackage{float}
\usepackage{subcaption}
\usepackage{enumerate}

\newtheorem{theorem}{Theorem}[section]
\newtheorem{proposition}[theorem]{Proposition}
\newtheorem{remark}[theorem]{Remark}
\newtheorem{lemma}[theorem]{Lemma}

\newcommand{\RR}{\mathbb{R}}

\newcommand{\bbI}{\mathbb{I}}

\newcommand{\cE}{\mathcal{E}}

\newcommand{\C}{\mathbb{C}}
\newcommand{\R}{\mathbb{R}}
\newcommand{\N}{\mathbb{N}}

\newcommand\be{{\bold e}}

\newcommand\bt{{\bold t}}

\def\cE{{\mathcal E}}
\def\cF{{\mathcal F}}
\def\cG{{\mathcal G}}

\def\cI{{\mathcal I}}
\def\cJ{{\mathcal J}}

\def\rd{{\mathrm{d}}}
\def\re{{\mathrm{e}}}

\newcommand\1{{\ensuremath {\mathds 1} }} 

\def\bra{\langle}
\def\ket{\rangle}

\def\Tr{{\rm Tr} \, }

\newcommand{\myfootnote}[1]{
	\renewcommand{\thefootnote}{}
	\footnotetext{\scriptsize#1}
	\renewcommand{\thefootnote}{\arabic{footnote}}
}

\title[Phase transition in the Peierls model]{Phase Transition in the Peierls model for polyacetylene}

\author{David Gontier \qquad Ad\'echola E. K. Kouande \qquad \'Eric S\'er\'e}

\date{\today}

\begin{document}
\myfootnote{David Gontier: CEREMADE, Université Paris-Dauphine, PSL University,75016 Paris, France \& ENS/PSL University, DMA, F-75005, Paris, France;\\
email: \href{gontier@ceremade.dauphine.fr}{gontier@ceremade.dauphine.fr}}
\myfootnote{Adéchola E. K. Kouande: CEREMADE, University of Paris-Dauphine, PSL University, 75016 Paris, France\\
email: \href{kouande@ceremade.dauphine.fr}{kouande@ceremade.dauphine.fr}}
\myfootnote{Éric Séré: CEREMADE, University of Paris-Dauphine, PSL University, 75016 Paris, France\\
email: \href{sere@ceremade.dauphine.fr}{sere@ceremade.dauphine.fr}}


\begin{abstract}
  We consider the Peierls model for closed polyactetylene chains with an even number of carbon atoms as well as infinite chains, in the presence of temperature. We prove the existence of a critical temperature below which the chain is dimerized, and above which it is $1$-periodic. The chain behaves like an insulator below the critical temperature and like a metal above it. We characterize the critical temperature in the thermodynamic limit model, and prove that it is exponentially small in the rigidity of the chain. We study the phase transition around this critical temperature.
    \bigskip
    
    \noindent \sl \copyright~2022 by the authors. This paper may be reproduced, in its entirety, for non-commercial purposes.
\end{abstract}

\maketitle    

\tableofcontents


\section{Introduction}

It is a well known fact that in closed polyacetylene molecular chains having an even number of carbon atoms ({\em e.g.} benzene), the valence electrons arrange themselves one link in two. This phenomenon is well understood in the Peierls model, introduced in $1930$ (see~\cite[p.108]{peierls1996quantum} and \cite{frohlich1954theory}), which is a simple non-linear functional describing, in particular, polyacetylene chains. This model is invariant under $1$-translations, but there is a symmetry breaking: the minimizers are dimerized, in the sense that they are $2$-periodic, but not $1$-periodic. This is known as \textbf{Peierls instability} or {\bf Peierls distortion} and is responsible for the high diamagnetism and low conductivity of certain materials such as bismuth~\cite{jones1934applications}.

In this paper, we study the Peierls model with temperature, and describe the corresponding phase diagram. We prove the existence of a critical temperature below which the chain is dimerized, and above which the chain is $1$-periodic. We characterize this critical temperature, and study the transition around it. In order to state our main results, let us first recall what is known for the Peierls model without temperature.


\subsection{The Peierls model at null temperature}
We focus on the case of even chains: We consider a periodic linear chain with $L=2N$ classical atoms (for an integer $N \ge 2$), together with quantum non-interacting electrons. We denote by $t_i$ the distance between the $i$-th and $(i + 1)$-th atoms, and set $\{ \bt \}:=~\{ t_1, \cdots, t_L\}$. By periodicity, we mean that the atoms indices are taken modulo $L$. The electrons are represented by a one-body density matrix $\gamma$, which is a self-adjoint operator on $\ell^2(\C^L)$, satisfying the Pauli principle $0 \le \gamma \le 1$. In this simple model, the electrons can hop between nearest-neighbour atoms, and feel a Hamiltonian of the form
\begin{equation}\label{CH1_3}
    T = T( \{ \bt\} ) := \begin{pmatrix}
        0      & t_1    & 0      & 0         & \cdots   & t_{L} \\
        t_1    & 0      & t_2    & \cdots    & 0        & 0       \\
        0      & t_2    & 0      & t_3       & \cdots   & 0        \\
        \vdots & \vdots & \vdots & \ddots    & \vdots   & \vdots    \\
        0      & 0      & \cdots &  t_{L-2} & 0        & t_{L-1}   \\
        t_{L} & 0      & \cdots & 0         & t_{L-1} & 0
    \end{pmatrix}.
\end{equation}
The Peierls energy of such a system reads~\cite{kennedy2004proof, lieb1961two, macris1996flux,peierls1996quantum, su1979solitons}
$$
    \widetilde{\cE}_{\rm full}^{(L)}(\{ \tilde{\bt} \}, \gamma) := \frac{1}{2}g \sum_{i=1}^{L}(\tilde{t_i} - b)^2 + 2 \Tr (T \gamma).
$$
The first term is the distortion energy of the atoms. Here, $b > 0$ is the equilibrium distance between two atoms and $g > 0$ is the rigidity of the chain. The second term models the electronic energy of the valence electrons (the $2$ factor stands for the spin). By scaling, setting $\tilde{t_i} = b t_i$ and $\mu = gb$, we have $\widetilde{\cE}_{\rm full}^{(L)}(\{ \tilde{t} \}, \gamma) =  
b \cE_{\rm full}^{(L)}(\{ t\}, \gamma)$,
with the energy
\begin{equation} \label{eq:energy_t_gamma_null_temperature}
    \boxed{ \cE^{(L)}_{\rm full}(\{ \bt \}, \gamma) := \frac{\mu}{2} \sum_{i=1}^{L}(t_i - 1)^2 + 2 \Tr (T \gamma). }
\end{equation}
There is only one parameter in the model, which is the strength $\mu > 0 $. In the so-called {\em half-filled} model, this energy is minimized over all $t_i > 0$ and all one-body density matrices (there is no constraint on the number of electrons):
\begin{equation*} 
    \boxed{ E^{(L)} := \min \left\{ \cE^{(L)}_{\rm full}(\{ \bt\}, \gamma), \quad  \bt \in \R_+^L, \quad 0 \le \gamma = \gamma^* \le 1 \right\}.}
\end{equation*}
One can perform the minimization in $\gamma$ first. We get
\begin{equation} \label{eq:min_gamma_theta=0}
    \min_{0 \le \gamma = \gamma^* \le 1} 2 \Tr \left( T \gamma  \right) = 2 \Tr \left( T \1 (T < 0 )\right) = - \Tr( | T | ) = - \Tr \left( \sqrt{T^2} \right),
\end{equation}
where we used here that $T$ is unitarily equivalent to $-T$, so that its spectrum is symmetric with respect to the origin. The optimal density matrix in this case is $\gamma_* = \1(T < 0)$, which has $\Tr(\gamma_* ) = N$ electrons (hence the denomination {\em half-filled}). The energy simplifies into
\begin{equation*} 
    E^{(L)} = \min \left\{ \cE^{(L)}(\{ \bt \}), \quad \bt \in \R_+^L \right\}, \quad \text{with} \quad \cE^{(L)}(\{ \bt \}) := \frac{\mu}{2} \sum_{i=1}^{L}(t_i - 1)^2 -\Tr(\sqrt{T^2}).
\end{equation*}
The energy $\cE^{(L)}$ only depends on $\{ \bt \}$, and is translationally invariant, in the sense that $\cE^{(L)}(\{\bt\}) = \cE^{(L)}(\{\tau_k \bt\})$ where $\{\tau_k \bt\} := \{ t_{k+1}, \cdots, t_{k+L} \}$. However, the minimizers of this energy are usually $2$-periodic, as proved by Kennedy and Lieb~\cite{kennedy2004proof} and Lieb and Nachtergaele~\cite{lieb1995stability}. More specifically, they proved the following: \\

\underline{Case $L \equiv 0 \mod 4$.} There are two minimizing configurations for $E^{(2N)}$, of the form 
    \begin{equation}\label{CH1_8}
        t_i = W + (-1)^i\delta \mbox{ or } t_i = W - (-1)^i\delta, \qquad \text{with $\delta > 0$.}
    \end{equation}

These two configurations are called {\bf dimerized} configurations \cite{kivelson1982hubbard}: they are $2$-periodic but not $1$-periodic. In other words, it is energetically favorable for the chain to break the $1$-periodicity of the model. We prove in Appendix~\ref{gain_of_energy} that the corresponding gain of energy is actually exponentially small in the limit $\mu \to \infty$.\\

\underline{Case $L \equiv 2 \mod 4$.} 
This case is similar, but we may have $\delta = 0$ for small values of $L$, or large values of $\mu$ (see also~\cite{kivelson1982hubbard}). There is $0 < \mu_c(L) < \infty$ so that, for $0 < \mu < \mu_c(L)$, there are still two dimerized minimizers, as in~\eqref{CH1_8}, while for $\mu > \mu_c(L)$, there is only one minimizer, which is $1$-periodic, that is $\delta = 0$.

\medskip

In all cases (with $L$ even), one can restrict the minimization problem to configurations $\{ \bt \}$ of the form $t_i = W \pm (-1)^i \delta$, and obtain a minimization problem with only two parameters.

\medskip

    Although $L$ is always even in the present paper, let us mention that molecules with $L$ odd and very large have been studied at zero temperature by Garcia Arroyo and S\'er\'e~\cite{garciaarroyo-sere}. In that case one gets ``kink solutions" in the limit $L\to\infty$.


\subsection{The Peierls model with temperature, main results.}
In the present article, we extend the results in the positive temperature case by modifying the Peierls model in order to take the entropy of the electrons into account. We denote by $\theta$ the temperature (the letter $T$ is reserved for the matrix in~\eqref{CH1_3}). Following the general scheme described in \cite[Section 4]{BLS}, the free energy is now given by (compare with~\eqref{eq:energy_t_gamma_null_temperature})
\begin{equation}\label{eq:F_full_theta_t_gamma}
    \boxed{ \cF_{{\rm full}, \theta}^{(L)}(\{ \bt \}, \gamma) := \frac{\mu}{2} \sum_{i=1}^{L}(t_i - 1)^2 + 2 \left\{ \Tr (T \gamma) + \theta \Tr (S(\gamma)) \right\}},
\end{equation}
with $S(x) := x \ln(x) + (1 - x) \ln(1-x)$ the usual entropy function. We consider again the minimization over all one-body density matrices, and study the minimization problem
\begin{equation*}
    \boxed{ F^{(L)}_\theta := \min \left\{ \cF_{{\rm full}, \theta}^{(L)}(\{ \bt\}, \gamma), \quad \bt \in \R_+^L, \quad 0 \le \gamma = \gamma^* \le 1 \right\}.}
\end{equation*}
There are now two parameters in the model, namely $\mu$ and $\theta$. The main goal of the paper is to study the phase diagram in the $(\mu, \theta)$ plane.

 As in~\eqref{eq:min_gamma_theta=0}, one can perform the minimization in $\gamma$ first (see Section~\ref{ssec:proof:L1_10} for the proof).
\begin{lemma}\label{L1_10}
We have 
\begin{equation}\label{CH1_49}
        \min_{0 \le \gamma \le 1} 2 \left\{ \Tr (T \gamma) + \theta\Tr (S(\gamma))\right\} = 
        - \Tr \left( h_\theta(T^2)  \right),
\end{equation}
with the function
\[
    h_\theta(x) := 2 \theta \,\ln \left( 2 \cosh \left( \frac{\sqrt{x}}{2 \theta}  \right)  \right).
\]
The minimization problem in the l.h.s of~\eqref{CH1_49} has the unique minimizer $\gamma_* = (1+\re^{ T/\theta})^{-1}$.
\end{lemma}

The properties of the function $h_\theta$ is given below in Proposition~\ref{prop:htheta}. The free Peierls energy therefore simplifies into a minimization problem in $\{ \bt \}$ only:
\begin{equation} \label{eq:free_energy}
     F^{(L)}_\theta = \inf \left\{ \cF^{(L)}_\theta( \{\bt\} ), \ \bt \in \R^L_+  \right\},
     \quad \text{with} \quad
    \cF^{(L)}_\theta( \{\bt\}) := \frac{\mu}{2} \sum_{i=1}^L (t_i - 1)^2 - \Tr \left( h_\theta(T^2)  \right).
\end{equation}

Our first theorem states that minimizers are always $2$-periodic, and that they become $1$-periodic when the temperature is large enough (phase transition).
\begin{theorem}\label{th:main} ~
For any $L=2N$, with $N$ an integer and $N \ge 2 $, there exists a critical temperature $\theta_c^{(L)}:=\theta_c^{(L)}(\mu) \ge 0$ such that:
\begin{itemize}
    \item for $\theta \ge \theta_c^{(L)}$, the minimizer of $\cF_\theta^{(L)}$ is unique and $1$-periodic;
    \item for $\theta \in (0,  \theta_c^{(L)})$ (this set is empty if $\theta_c^{(L)} = 0$), there are exactly two minimizers, which are dimerized, of the form~\eqref{CH1_8}.
\end{itemize}
In addition,
\begin{enumerate}[i)]
\item If $L\equiv0\mod4$, this critical temperature is positive ($\theta_c^{(L)}(\mu) > 0$ for all $\mu > 0$).
\item If $L\equiv 2\mod 4$, there is $\mu_c:= \mu_c(L) > 0$ such that for $\mu \le \mu_c$, $\theta_c^{(L)} $ is positive ($\theta_c^{(L)} >0$),  whereas for $\mu > \mu_c$, $\theta_c^{(L)} = 0$. Moreover as a function of $L$ we have $\mu_c(L) \sim \frac{2}{\pi}\ln(L)$ at $+\infty$.
\end{enumerate}
\end{theorem}

This theorem only deals with an even number $L$ of atoms. One expects a similar behaviour for $L$ odd and large, but the arguments in the proof are not sufficient to guarantee this: they only imply that the minimizer is one-periodic when the temperature is large enough (see Remark \ref{odd}). We do not know what exactly happens for a small positive temperature and an odd number $L$.
\medskip

We postpone the proof of Theorem~\ref{th:main} until Section~\ref{sec:Peierls_T}. The first part uses the concavity of the function $h_\theta$ on $\RR_+$, while those of~$i)$ and~$ii)$ are based on the Euler-Lagrange equations.\\

\medskip

As in the null temperature case, minimizers are always $2$-periodic, hence the minimization problem is a minimization over the two variables $W$ and $\delta$. Actually, we have
\[
    F_\theta^{(2N)} = (2N) \min \left\{ g_\theta^{(2N)}(W, \delta), \quad W \ge 0, \ \delta \ge 0\right\},
\]
with the energy {\em per unit atom} (the following expression is justified below in Eqn.~\eqref{eq:gthetaL})
\begin{equation} \label{eq:def:gthetaL}
    g_\theta^{(2N)}(W, \delta) = \frac{\mu }{2} \left[ (W- 1)^2 + \delta^2 \right] -  \frac{1}{2N} \sum_{k=1}^{2N} h_\theta  \left( 4 W^2 \cos^2 \left( \frac{2 k \pi}{2N} \right) + 4 \delta^2 \sin^2 \left(  \frac{2 k \pi}{2N}  \right) \right).
\end{equation}
We recognize a Riemann sum in the last expression. This suggests that we can take the thermodynamic limit $L \to \infty$. This limit is quite standard in the physics literature on long polymers: many theoretical papers present models of polymers at null temperature that are directly written for infinite chains (see {\it e.g} \cite{su1979solitons}).\medskip

We define the thermodynamic limit free energy (per unit atom) as
\begin{equation}\label{thermo-energy}
    \boxed{ \displaystyle f_\theta := \displaystyle \liminf _{N \rightarrow + \infty} \frac{1}{2N}F_\theta^{(2N)}}.
\end{equation}
As expected, we have the following (see Section~\ref{sec:proof:justification_thermo} for the proof).
\begin{lemma}\label{lem:justification_thermo}
We have $f_\theta = \min \left\{ g_\theta(W, \delta), \quad W \ge 0, \ \delta \ge 0 \right\}$ with
    \[
     g_\theta(W, \delta) := \frac{\mu }{2} \left[ (W- 1)^2 + \delta^2 \right] - \frac{1}{2 \pi}\int_0^{2 \pi} h_\theta  \left( 4 W^2 \cos^2 ( s ) + 4 \delta^2 \sin^2 (s) \right) \rd s.
    \]
\end{lemma}

The next theorem is similar to Theorem~\ref{th:main}, and shows the existence of a critical temperature for the thermodynamic model. Its proof is postponed until Section~\ref{sec:proof:main_thermo}, and is based on the study of the Euler-Lagrange equations.\\

\begin{theorem}\label{th:main_thermo}~
There is a critical (thermodynamic) temperature $\theta_c = \theta_c(\mu) > 0$, which is always positive, and so that for all $\theta \ge \theta_c,$ the minimizer of $g_\theta$ satisfies $\delta = 0$, whereas for all $\theta < \theta_c$, it satisfies $\delta > 0$. \\
In the large $\mu$ limit, we have
\[
     \theta_c(\mu) \sim C\exp \left( - \frac{\pi}{4} \mu \right), \quad \text{with} \quad C \approx 0.61385.
\]
\end{theorem}

This reflects the fact that for an infinite chain, there is a transition between the dimerized states ($\delta > 0$), which are insulating (actually, one can show that the gap of the $T$ matrix is of order $\delta$), to the $1$-periodic state (with $\delta = 0$), which is metallic, as the temperature increases. This can be interpreted as an insulating/metallic transition for polyacetylene. Such a phase transition has been observed experimentally in the blue bronze in~\cite{pouget1983evidence}. We display in Figure~\ref{fig:fig} (left) the map $\mu \mapsto \theta_c(\mu)$ in the $(\mu, \theta)$ plane.

\medskip

In~\eqref{thermo-energy}, we only consider the limit $L = 2N \to \infty$ to define the thermodynamic critical temperature $\theta_c$. Note that the cases $L\equiv0\mod4$ and $L\equiv2\mod4$ merge when $L$ tends to infinity: this is consistent with the fact that the critical stiffness $\mu_c(L)$  tends to infinity as $L\to\infty$ in Theorem \ref{th:main}. We also expect odd chains to behave like even chains, but the study of the odd case is more delicate since we do not have an analogue of~\eqref{eq:def:gthetaL} and we leave it for future work.

\medskip

Finally, we study the nature of the transition. It is not difficult to see that $\delta \to 0$ as $\theta \to \theta_c$. Actually, there is a bifurcation around this critical temperature, see also Figure~\ref{fig:fig} (right).
\begin{theorem}\label{th:bifurcation}
There is $C > 0,$ such that $\delta(\theta) = C\sqrt{(\theta_c -  \theta)_+} + o\left(\sqrt{(\theta_c -  \theta)_+}\right).$
\end{theorem}

We postpone the proof of Theorem~\ref{th:bifurcation} until Section~\ref{bifurcation}. It mainly uses the implicit function theorem. The value of $C$ is explicit and is given in the proof.

\begin{figure}[!ht]
    \begin{subfigure}[b]{0.4\textwidth}
         \centering
         \includegraphics[scale = 0.4]{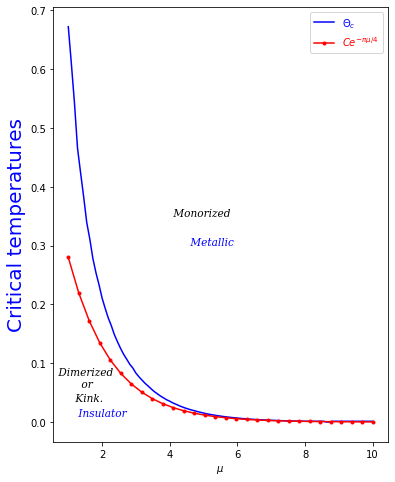}
     \end{subfigure}
 \begin{subfigure}[b]{0.5\textwidth}
         \centering
         \includegraphics[scale = 0.4]{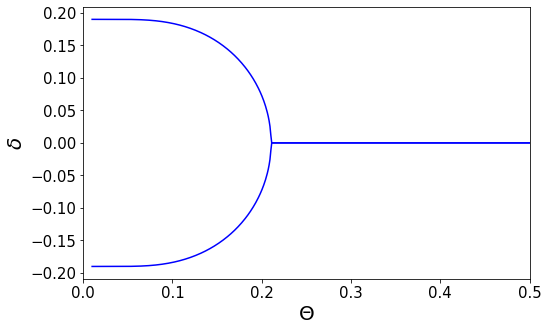}
\label{fig:bifurcation}
     \end{subfigure}
     \caption{Numerical simulations. (Left) the critical temperature $\mu \mapsto \theta_c(\mu)$ and its asymptotic $Ce^{-\frac{\pi}{4}\mu}$. (Right) The bifurcation of $\delta$ in the thermodynamic model. We took $\mu = 2$, and the critical temperature is found to be $\theta_c = 0.2112$.}
     \label{fig:fig}
\end{figure}

\section{Proofs in the finite chain Peierls model with temperature}
\label{sec:Peierls_T}

We now provide the proofs of our results. We gather in this section the proofs of the finite $L = 2N$ model, and postpone the ones of the thermodynamic model to the next Section. 

\subsection{Proof of Lemma~\ref{L1_10}, and properties of the $h$ functional}
\label{ssec:proof:L1_10}

First, we justify the functional $\cF^{(L)}_\theta$ appearing in~\eqref{eq:free_energy}, and provide the proof of Lemma~\ref{L1_10}.
\begin{proof}
 We study the minimization problem
\[
    \min_{0 \le \gamma \le 1} 2 \left\{ \Tr (T \gamma) + \theta \Tr (S(\gamma) ) \right\}.
\]
Any critical point $\gamma^*$ of the functional satisfies the Euler-Lagrange equation
\begin{equation} \label{eq:EL_gammastar}
    T + \theta S'(\gamma_*) = 0, \quad \text{that is}  \quad T + \theta \ln \left( \frac{\gamma_*}{1 - \gamma_*} \right) = 0.
\end{equation}
There is therefore only one such critical point, given by
\[
\gamma_* =  \frac{1}{1 + \re^{ T/\theta}} 
= \frac{\re^{-T/(2 \theta)}}{2 \cosh(T/(2 \theta))}, \quad \text{hence} \quad
1 - \gamma_* =  \frac{1}{1 + \re^{ -T/\theta}} = \frac{\re^{T/(2 \theta)}}{2 \cosh(T/(2 \theta))}.
\]
By convexity of the functional, this critical point is the minimizer. For this one-body density matrix, we obtain, using~\eqref{eq:EL_gammastar}
\begin{align*}
    2  \left\{ \Tr (T \gamma_*) + \theta \Tr (S(\gamma_*) ) \right\} & = 2 \Tr \left(\gamma_* \left[ T + \theta \ln \left(\frac{\gamma_*}{1 - \gamma_*}\right) \right] + \theta \ln (1 - \gamma_* )\right) \\
    & = 2 \theta \Tr (\ln(1 - \gamma_*)) =  2 \theta \Tr \left(T/2 \theta\right) - 2 \theta \Tr \left( \ln \left[ 2 \cosh(T/2\theta) \right] \right).
\end{align*}
Finally, since $T$ is unitary equivalent to $-T$, we have $\Tr(T) = 0$. This gives as wanted
\[
     \min_{0 \le \gamma \le 1} 2 \left\{ \Tr (T \gamma) + \theta \Tr (S(\gamma) ) \right\} = - \Tr \left( h_\theta(T^2)  \right), \quad \text{with} \quad h_\theta(x) := 2 \theta \ln \left( 2 \cosh \left( \frac{ \sqrt{x}}{2\theta} \right) \right).
\]
\end{proof}

Let us gather here some properties of the function $h_\theta$, that we will use throughout the article.
\begin{proposition} \label{prop:htheta}
    We have $h_\theta(x) = \theta h \left( \frac{x}{4 \theta^2} \right)$ and $h'_\theta(x) = \frac{1}{4 \theta} h' \left( \frac{x}{4 \theta^2} \right)$, with
    \[
        h(y) = 2 \ln(2 \cosh( \sqrt{y} )), \quad \text{and} \quad
        h'(y) = \dfrac{\tanh(\sqrt{y})}{\sqrt{y}}.
    \]
    In particular, $h$ (hence $h_\theta$) is positive, increasing and concave. We have $\lim_{y \to 0} h'(y) = 1$, and the inequality $h_\theta(x) \ge \sqrt{x}$, valid for all $\theta > 0$ and all $x \ge 0$. In addition, we have the pointwise convergence $h_\theta(x) \to \sqrt{x}$ as $\theta \to 0$.
\end{proposition}
The last part shows that we recover the model at zero temperature. The concavity of $h$ comes from the fact that $h'$ is positive and decreasing. Another way to see concavity is that $h_\theta(t) = \min_{0 \le g \le 1} 2 \{ t g + \theta S(g) \}$ is the minimum of linear functions (in $t$), hence concave. The inequality $h_\theta(x) \ge \sqrt{x}$ comes from $2 \cosh(x) \ge \re^x$. 

\subsection{Proof of Theorem~\ref{th:main}: Existence of a critical temperature}
\label{Critical_T}

We now study the minimizers of $\cF^{(L)}_\theta(\{ \bt\})$ in~\eqref{eq:free_energy}, which we recall is given by
\[
   \cF^{(L)}_\theta(\{ \bt\}) :=  \frac{\mu}{2} \sum_{i=1}^L (t_i - 1)^2 - \Tr \left( h_\theta(T^2)  \right).
\] 
First, we prove that the minimizers are always $2$-periodic. We then study the existence of a critical temperature. For the first part, our strategy follows closely the argument of Kennedy and Lieb in~\cite{kennedy2004proof}, and relies on the concavity of $h_\theta$. 

\medskip

\underline{All minimizers are $2$-periodic}. Recall that if $x \mapsto \varphi(x)$ is concave over $\R_+$, then $A \mapsto \Tr (\varphi(A))$ is concave over the set of positive matrices. Applying this property to $h_\theta$ which is concave on $\R_+$, we have
\[
   \Tr( h_\theta(T^2)) \le \Tr(h_\theta ( \bra T^2 \ket )),
\]
where $\bra T^2 \ket$ is defined as in \cite{kennedy2004proof} as the average of $T^2$ over all translations:
\begin{equation*}
    \bra T^2 \ket = \frac{1}{L}\sum_{k=1}^L \Theta_k T^2 \Theta_k^{-1}, \mbox{ with } \Theta_k = \Theta_1^k \mbox{ and } \Theta_1:= \begin{pmatrix}
0 & 1 & 0 & \cdots  & 0 \\
0 & 0 & 1 & \cdots  & 0 \\
\vdots & \vdots & \vdots & \ddots & \vdots \\
0 & 0 & 0 & \cdots &  1 \\
1 & 0 & 0 & \cdots & 0 
\end{pmatrix}.
\end{equation*}
This implies the lower bound 
\begin{equation} \label{eq:lower_bound_FG}
    F^{(L)}_\theta \ge G_\theta^{(L)}
\end{equation}
where
\begin{equation} \label{MOY_ENERGY}
     G^{(L)}_\theta := \inf \left\{ \cG^{(L)}_\theta( \{\bt\} ), \ \bt \in \R^L_+  \right\},
     \quad \text{with} \quad
    \cG_\theta^{(L)}(\{\bt\}) = \frac{\mu}{2} \sum_{i=1}^L (t_i - 1)^2 - \Tr \left( h_\theta(\bra T^2 \ket )  \right).
\end{equation}
In addition, we have equality in~\eqref{eq:lower_bound_FG} iff the optimal $\{\bt\}$ for $G_\theta^{(L)}$ satisfies $T(\{ \bt \})^2 = \bra T(\{ \bt \})^2 \ket$. Note that
\[
    T^2 = \begin{pmatrix}
        t_L^2 + t_1^2   & 0    & t_1 t_2      & 0         & \cdots   & 0\\
        0    & t_1^2 + t_2^2  & 0    & t_2 t_3 & \cdots         & t_L t_1       \\
        t_1 t_2      & 0    & t_2^2 + t_3^2      & 0       & \cdots   & 0        \\
        \vdots & \vdots & \vdots & \ddots    & \vdots   & \vdots    \\
        t_{L-1} t_L      & 0      & \cdots &  0 & t_{L-2}^2 + t_{L-1}^2        & 0   \\
        0 & t_L t_1      & \cdots & t_{L-1} t_L         & 0 & t_{L}^2 + t_1^2
    \end{pmatrix}.
\]
So we have $T(\{ \bt \})^2 = \bra T(\{ \bt \})^2 \ket$ iff $t_i^2 + t_{i+1}^2$ and $t_i t_{i+1}$ are independent of $i$. This happens only if $T$ is $2$-periodic.

Introducing the variables (our notation slightly differ from the ones in~\cite{kennedy2004proof}: we put $z^2$ instead of $z$, so that all quantities $(x,y,z)$ are homogeneous)
\begin{equation*} 
    x := \frac1L \sum_{i=1}^L t_i, \quad y^2 := \frac1L \sum_{i=1}^L t_i^2, \quad z^2 = \frac{1}{L} \sum_{i=1}^L t_i t_{i+1},
\end{equation*}
we obtain $\bra T^2 \ket = 2 y^2 \bbI_L +  z^2 \Omega_L$ with $\Omega_L := \Theta_2 + \Theta_2^*$, and
\begin{equation*} 
 \cG^{(L)}_\theta(\{\bt\}) =  \widetilde{\cG}^{(L)}_\theta(x,y,z) :=  \frac{\mu L}{2}(y^2 - 2x + 1) - \Tr \left( h_\theta(2y^2 \bbI_L + z^2 \Omega_L) \right) .
\end{equation*}
The function $\widetilde{\cG}^{(L)}_\theta$ is much easier to study, as it only depends on the three variables $(x,y,z)$. Let us identify the triplets $(x,y,z)$ coming from a $2$-periodic or $1$-periodic state.

\begin{lemma} ~
    \begin{itemize}
        \item For all $\bt \in \R^L_+$, the corresponding triplet $(x,y,z)$ belongs to 
        \[
        X := \left\{ (x, y, z) \in \R^3_+, \quad  y^2 \ge x^2, \quad z^2 \ge \max \{ 0, 2 x^2 - y^2 \} \right\}.
        \]
        \item If $L = 2N$ is even, the configuration $\bt$ is $2$-periodic of the form~\eqref{CH1_8} iff the triple $(x,y,z)$ belongs to
        \[
            X_2 := \left\{ (x, y, z) \in \R^3_+ \ \text{of the form} \ x = W, \ y^2 = W^2 + \delta^2, \ z^2 = W^2 - \delta^2 \right\}.
        \]
        This happens iff $z^2 = 2x^2 - y^2$.
        \item The configuration $\bt$ is $1$-periodic, of the form $\bt = (W, \cdots, W)$ iff $(x,y,z)$ belongs to
        \[
            X_1 :=  \left\{ (x, y, z) \in \R^3_+ \ \text{of the form} \ x = y = z = W \right\}.
        \]
        This happens iff $z^2 = 2x^2 - y^2$ and $x = y$.
    \end{itemize}
    
\end{lemma}
\begin{proof}
    By Cauchy-Schwarz, we have
    \[
        x^2 = \frac{1}{L^2} \left( \sum_{i=1}^L t_i \right)^2 \le \frac{1}{L} \sum_{i=1}^L t_i^2 = y^2,
    \]
    which is the first equality. Next, we have, 
    \[
        z^2 = \frac{1}{2 L} \sum_{i=1}^L \left[ (t_i + t_{i+1})^2 - t_i^2 - t_{i+1}^2 \right] =  \frac{1}{2 L} \sum_{i=1}^L (t_i + t_{i+1})^2 - y^2.
    \]
    On the other hand, we have by Cauchy-Schwarz,
    \[
        x^2 = \left( \frac{1}{2 L} \sum_{i=1}^L (t_i + t_{i+1}) \right)^2 \le \frac{1}{4 L} \sum_{i=1}^L (t_i + t_{i+1})^2.
    \]
    This proves that $z^2 \ge 2 x^2 - y^2$. The other parts of the Lemma can be easily checked.
\end{proof}

\begin{lemma}\label{mini_Gover_X}
    For any integer $L > 2$ and all $\theta \ge 0$, the minimizers of $\widetilde{\cG}^{(L)}_\theta$ over $X$ belong to $X_2$.
\end{lemma}

\begin{proof}
    Let us fix $x$ and $y$, and look at the minimization over the variable $z$ only. Setting $Z := z^2$, we see that
    \[
        \varphi: Z \mapsto \Tr \left( h_\theta (2 y^2 \bbI_L + Z \Omega_L)  \right)
    \]
    is concave. In addition, the derivative of $\varphi$ at $Z = 0$ equals
    \[
        \varphi'(Z) =  \Tr \left( h_\theta' (2 y^2 ) \Omega_L  \right) = h_\theta' (2 y^2 ) \Tr ( \Omega_L ) = 0,
    \]
    where we used that $\Omega_L$ only has null elements on its diagonal. We deduce that $\varphi$ is decreasing on $\R_+$. So the minimizer of $\widetilde{\cG}^{(L)}_\theta$ must saturate the lower bound constraint $z^2 = \max \{ 0, 2x^2 - y^2\}$. \\
    
    We now claim that the optimal triplet $(x,y,z)$ satisfies $2x^2 - y^2 \ge 0$. Assume otherwise that $2x^2 - y^2 < 0$, hence $z^2 = 0$. We have
    \[
        \widetilde{\cG}^{(L)}_\theta(x,y,0) = \frac{\mu L}{2}(y^2 - 2x + 1) - \Tr \left( h_\theta (2y^2) \right) = L \left(  \frac{\mu}{2} (y^2 - 2x + 1) - h_\theta(2 y^2) \right).
    \]
    This function is decreasing in $x$, so the optimal $x$ saturates the constraint $x^2 = y^2$. But in this case, we have $2x^2 - y^2 = y^2 \ge 0 $, a contradiction. This proves that, for the optimizer, we have $2 x^2 - y^2 \ge 0$, and $z^2 = 2x^2 - y^2$. Finally, $(x,y,z)$ belongs in $X_2$.
\end{proof}
Let $(x_*, y_*, z_*) \in X_2$ be the minimizer of $\widetilde{\cG}_\theta^{(L)}$, and let $W \ge 0$ and $\delta \ge 0$ be so that $x_* = W$, $y_*^2 = W^2 + \delta^2$, and $z_* = W^2 - \delta^2$. Let $\bt_*$ be one of the the two $2$-periodic states $W \pm (-1)^i \delta$. We have $T(\{ \bt_* \})^2 = \bra T(\{ \bt_* \})^2 \ket$, which leads to the chain of inequalities
\[
    F_\theta^{(L)} \ge G_\theta^{(L)} \ge \min_{(x, y,z)} \widetilde{\cG}_\theta^{(L)} = \widetilde{\cG}_\theta^{(L)}(x_*, y_*, z_*) = G_\theta^{(L)}(\{ \bt_* \}) = \cF_\theta( \{ \bt_* \}) \ge F_\theta^{(L)}.
\]
We therefore have equalities everywhere. Since only the $2$-periodic states $W \pm (-1)^i \delta$ gives the optimal triplet $(x_*, y_*, z_*)$, they are the only minimizers. This proves that all minimizer of $\cF_\theta^{(L)}$ are $2$-periodic. They are two dimerized minimizers if $\delta > 0$, and a unique $1$-periodic minimizer if $\delta = 0$.

\begin{remark}\label{odd}
In the case of odd chains, we still have the equation $F_\theta^{(L)}  \ge G_\theta^{(L)}$ in~\eqref{eq:lower_bound_FG}. However, the optimal triplet $(x_*, y_*, z_*)$ does not usually come from a state $\{ \bt_* \}$: an odd chain cannot be dimerized. It can however come from such a state if $\delta = 0$, that is if $\bt_*$ is actually one-periodic. One can therefore prove that also for odd chains, minimizers become $1$-periodic for large enough temperature.
\end{remark}

\medskip


\underline{Existence of the critical temperature.}
Since all minimizers are $2$-periodic, we can parametrize $\cG_\theta^{(L)}$ as a function of $(W, \delta)$ instead of $\{ \bt \}$. So we write (in what follows, we normalize by $L$ to get the energy per atom)
\[
    g^{(L)}_\theta(W, \delta) = \frac{\mu }{2} \left[ (W- 1)^2 + \delta^2 \right] - \frac{1}{L} \Tr \left( h_\theta ( 2 (W^2 + \delta^2) \bbI_L  + (W^2 - \delta^2) \Omega_L )  \right).
\]
To compute the last trace, we compute the spectrum of $\Omega_L$. We have, for all $1 \le k \le L$,
\[
    \Omega_L \be_k = 2 \cos \left(  \frac{4 k \pi}{L} \right) \be_k, \quad \text{where} \quad \be_ k =  (1, \re^{2i \pi k/L}, \re^{2 \cdot 2i \pi k /L}, \cdots, \re^{(L-1) \cdot 2i \pi k /L})^T.
\]
So 
\[
    \sigma \left(  \Omega_L \right) := \left\{ 2 \cos \left( \frac{4 k \pi}{L} \right) , \quad 1 \le k \le L \right\}.
\]
This shows that
\begin{align}
   g^{(L)}_\theta(W, \delta) & = \frac{\mu }{2} \left[ (W- 1)^2 + \delta^2 \right] -  \frac{1}{L}\sum_{k=1}^L h_\theta  \left( 2 (W^2 + \delta^2)  +2 (W^2 - \delta^2) \cos \left( \frac{4 k \pi}{L} \right) \right) \nonumber \\
    & = \frac{\mu }{2} \left[ (W- 1)^2 + \delta^2 \right] -  \frac{1}{L} \sum_{k=1}^L h_\theta  \left( 4 W^2 \cos^2 \left( \frac{2 k \pi}{L} \right) + 4 \delta^2 \sin^2 \left(  \frac{2 k \pi}{L}  \right) \right), \label{eq:gthetaL}
\end{align}
which is the expression given in~\eqref{eq:def:gthetaL}. The function $g_\theta$ appearing in Lemma~\ref{lem:justification_thermo} has a similar expression, but we replace the last Riemann sum by the corresponding integral.

First, we prove that for $\theta$ large enough, the minimizer is $1$-periodic (corresponding to $\delta = 0$).
\begin{lemma} \label{lem:existence_thetac_G}
    For all $\theta \ge \frac{1}{\mu}$, the minimizer of $\cG_\theta^{(L)}$ satisfies $\delta = 0$. The same holds for the function $g_\theta$ (thermodynamic limit case).
\end{lemma}

\begin{proof}
  We prove the result in the thermodynamic limit, but the proof works similarly at fixed $L$. Let $(W_1, 0)$ denote the minimizer of $g_\theta$ among $1$-periodic configurations (that is with the extra constraint that $\delta = 0$). Writing that $\partial_W g_\theta (W_1, 0)= 0$, we obtain that
\begin{equation}\label{EL}
    \mu (W_1 - 1) =
     \frac{W_1}{ \pi \theta} \int_0^{2 \pi} h'\left( \dfrac{W_1^2 \cos^2(s)}{\theta^2}\right)  \cos^2(s)  \rd s.
\end{equation}
For any other configurations $(W, \delta)$, we write $W = W_1 + \varepsilon$, and obtain that
\begin{align*}
     g_\theta (W_1 + \varepsilon, \delta) - g_\theta(W_1, 0) 
     & = \frac{\mu}{2} \left[ 2(W_1 - 1) \varepsilon + \varepsilon^2 + \delta^2 \right] \\
     & \quad   - \frac{\theta}{2 \pi} \int_0^{2 \pi} \left[ h \left(  \dfrac{ (W_1 + \varepsilon)^2 \cos^2(s) + \delta^2 \sin^2(s)}{\theta^2 }\right) - h \left(  \dfrac{W_1^2 \cos^2 (s)}{\theta^2}  \right) \right] \rd s.
\end{align*}
Using that $h$ is concave, we have $h(a + b) - h(a) \le h'(a)b$, so, with $a = W_1^2 \cos^2(s) / \theta^2$ and $b = \left[ \delta^2 \sin^2(s) + (2W_1\varepsilon + \varepsilon^2) \cos^2( s) \right]/\theta^2$, we get
\begin{align*}
    &g_\theta (W_1 + \varepsilon, \delta) - g_\theta(W_1, 0) \\
    & \ge \mu(W_1 - 1) \varepsilon + \frac{\mu}{2} \varepsilon^2 + \frac{\mu}{2} \delta^2  \\
    & \quad  - \frac{1}{2 \pi \theta} \int_0^{2 \pi}  h'  \left( \dfrac{ W_1^2 \cos^2 (s)}{\theta^2}  \right) \left[ \delta^2 \sin^2(s) + (2W_1\varepsilon + \varepsilon^2) \cos^2( s) \right] \rd s.
\end{align*}
Using~\eqref{EL}, the term linear in $\varepsilon$ vanishes. In addition, since $h'' < 0$ on $\R_+$, we have $h'(x) \le~h'(0)=~1$. This gives
\begin{align*}
    & g_\theta (W_1 + \varepsilon, \delta) - g_\theta(W_1, 0) 
    \ge \left( \frac{\mu}{2} -  \frac1{2 \theta} \right) \varepsilon^2 +  \left( \frac{\mu}{2} -  \frac1{2 \theta} \right)  \delta^2 .
\end{align*}
The right-hand side is positive whenever $\theta >\frac{1}{\mu}$, which proves the result.
\end{proof}

In what follows, we define the critical temperature $\theta_c = \theta_c(\mu)$ by
\[
    \theta_c := \inf \{ \theta \in \R_+, \  \text{the minimizer of $g_{\theta'}$ has $\delta = 0$ for all $\theta' \ge \theta$} \}.
\]
We define similarly $\theta_c^{(L)} = \theta_c^{(L)}(\mu)$ for the case of finite chains.

\medskip

\underline{Study of the critical temperature in the case $L \in 2 \N$.}
We now study $\theta_c^{(L)}$ with $L = 2N$, $N \ge 2$, and prove that it is strictly positive if $L \equiv 0 \mod 4$, and that, if $L \equiv 2 \mod 4$, there is $\mu_c = \mu_c(L)$ so that $\theta_c^{(L)}(\mu) > 0$ iff $\mu < \mu_c(L)$.

For fixed $\theta$, any minimizing configuration $(W, \delta)$ satisfies the Euler-Lagrange equations
\begin{equation*}
    (\partial_W g_\theta^{(L)}, \partial_\delta g_\theta^{(L)})(W, \delta) = (0,0).
\end{equation*}
This gives the set of equations
\begin{equation} \label{ELEQ}
    \begin{cases}
		\mu(W - 1) & =  \displaystyle \frac{2W}{\theta} \frac{1}{L}\sum_{k = 1}^{L} h'\left(\frac{W^2}{\theta^2} \cos^2( \tfrac{2k\pi}{L}) + \frac{\delta^2}{\theta^2} \sin^2( \tfrac{2k\pi}{L}) \right)\cos^2(\tfrac{2k\pi}{L}) \\
		\mu \delta & =  \displaystyle  \frac{2\delta }{\theta } \frac{1}{L}\sum_{k = 1}^{L} h'\left(\frac{W^2}{\theta^2} \cos^2(\tfrac{2k\pi}{L}) + \frac{\delta^2}{\theta^2} \sin^2( \tfrac{2k\pi}{L})\right)\sin^2(\tfrac{2k\pi}{L}).
	\end{cases}
\end{equation}
Note that the second equation always admits the trivial solution $\delta = 0$. This corresponds to the critical point among $1$-periodic configurations. It is the unique solution if $\theta \ge \theta_c^{(L)}$, but for $\theta \in (0, \theta_c^{(L)})$, there are other critical points, corresponding to the dimerized configurations. Actually, as $\theta$ varies, we expect two branches of solutions: the branch of $1$-periodic configuration, and the branch of dimerized configurations. These two branches cross only at $\theta = \theta_c$ (see Figure~\ref{fig:fig} (right)).

\medskip

In order to focus on the branch of dimerized configurations, we factor out the $\delta$ factor in the second equation. Now, $\delta = 0$ is no longer a solution, unless we are exactly at the critical temperature $\theta_c^{(L)}$. So, in order to find this critical temperature, we seek the solution, in $(W, \theta)$, of (we multiply the second equation by $W$ for clarity)
\begin{equation} \label{Eul_Lag__L_fini}
\begin{cases}
		\mu(W - 1) & =  \displaystyle \frac{2W}{\theta} \frac{1}{L}\sum_{k = 1}^{L} h'\left(\frac{W^2}{\theta^2} \cos^2( \tfrac{2k\pi}{L})\right)\cos^2(\tfrac{2k\pi}{L}) \\
		\mu W & =  \displaystyle  \frac{2 W}{\theta } \frac{1}{L}\sum_{k = 1}^{L} h'\left(\frac{W^2}{\theta^2} \cos^2(\tfrac{2k\pi}{L})\right)\sin^2(\tfrac{2k\pi}{L}).
\end{cases}
\end{equation}


\begin{lemma}\label{lem:uniqueness_of_theta_c_L}
For all $\mu > 0$, there is a unique solution $(W, \theta)$ of \eqref{Eul_Lag__L_fini} in the case $L~=~0\mod4$, whereas if $L=2\mod4$, there is some value $\mu_c:=\mu_c(L)$ such that for all $\mu > \mu_c$, \eqref{Eul_Lag__L_fini} has no solution, and has a unique one if $\mu \le \mu_c$. Moreover in the last case $\mu_c(L)\sim \frac{2}{\pi}\ln(L)$ at $+\infty$.
\end{lemma}

\begin{proof}
    We write $L = 2N$, and note that the terms $k$ and $k + N$ gives the same contribution. Taking the difference of the second and first equations of \eqref{Eul_Lag__L_fini}, we obtain
	\[
	    \mu  =  - \displaystyle  \frac{2 W}{\theta } \frac{1}{N}\sum_{k = 1}^{N} h'\left(\frac{W^2}{\theta^2} \cos^2(\tfrac{k\pi}{N})\right)\cos(\tfrac{2k\pi}{N}).
	\]
	Recall that $h'(t) = \frac{\tanh(\sqrt{t})}{\sqrt{t}}$ for $t\neq 0$ and $h'(0) = 1$. The point $t = 0$ therefore plays a special role. The argument of $h'$ equals $0$ for $k =\frac{N}{2}$, which happens only if $N \equiv 0 \mod 2$ (that is $L \equiv 0 \mod 4$). In this case, the equation becomes, with $x := \frac{W}{\theta}$ (we write $L = 2N = 4n$)
	\begin{equation}\label{Eq_J_L}
	\mu = -\frac{1}{n}
	  \sum_{\tiny{\begin{matrix} 
			k = 1 \\
			k\neq n
			\end{matrix}} }^{2n}\frac{\tanh\left(x\cos(\frac{k\pi}{2n})\right)}{\cos(\frac{k\pi}{2n})}\cos(\tfrac{k\pi}{n}) + \frac{x}{n}=: \cJ_{2n}(x).
	\end{equation}
    The function $\cJ_{2n}$ is smooth. The first sum is uniformly bounded for $x \in \R_+$ while the second diverges, so $\cJ_{2n} = 0$ and $\cJ_{2n}(+ \infty) = + \infty$. We claim that $\cJ_{2n}$ is increasing. The intermediate value theorem then gives the existence and uniqueness of the solution of $\cJ_{2n}(x) = \mu $ on $\R_+$. This gives $\frac{W}{\theta} = \cJ_{2n}^{-1}(\mu)$. We then deduce respectively $\theta$ and $W$ from the first and second equations of~\eqref{Eul_Lag__L_fini}. This proves that~\eqref{Eul_Lag__L_fini} has a unique solution. The corresponding temperature is the critical temperature $\theta_c^{(L)}$. \\
    
    It remains to prove that $\cJ_{2n}$ is increasing. Splitting the sum in~\eqref{Eq_J_L} into $2$ sums of size $(n-1)$, we get
	\[
	\cJ_{2n}(x) = \frac{1}{n}\left(x - \tanh(x)\right) + \frac{1}{n}\sum_{k=1}^{n-1}\left(\frac{\tanh\left(x\sin\left(\frac{k\pi}{2n}\right)\right)}{\sin\left(\frac{k\pi}{2n}\right)} - \frac{\tanh\left(x\cos\left(\frac{k\pi}{2n}\right)\right)}{\cos\left(\frac{k\pi}{2n}\right)}\right)\cos\left(\tfrac{k\pi}{n}\right).
	\]
	Its derivative is given by 
	\[
	\cJ_{2n}'(x) = \frac{1}{n}\left(1 - \frac{1}{\cosh^2(x)}\right) +  \frac{1}{n} \sum_{k = 1}^{n-1}\left(\frac{1}{\cosh^2\left(x\sin\left(\frac{k\pi}{2n}\right) \right)} - \frac{1}{\cosh^2\left(x\cos\left(\frac{k\pi}{2n}\right) \right)} \right)\cos(\tfrac{k\pi}{n}).
	\]
	For all $s \in [0, 1]$, the function
	\[
	    \left [ \cosh^{-2}\left(x\sin\left(\tfrac{\pi}{2}s\right) \right) - \cosh^{-2}\left(x\cos\left(\tfrac{\pi}{2}s\right) \right) \right]\cos(\pi s)
	\]
	is positive (both terms are positive if $s \in [0, 1/2]$, and both are negative if $s \in [1/2, 1]$). This shows that $\cJ_{2n}$ is increasing as wanted.

    \medskip

	In the case $N \equiv 1 \mod 4$ (that is $L = 2 \mod 4$) the argument of $h'$ is never null, and we simply have (we write $L = 2N = 4n + 2$)
   \[
    \mu 
	  = -\frac{1}{2n+1}
	  \sum_{k = 1}^{2n+1}\frac{\tanh\left(x\cos(\frac{k\pi}{2n+1})\right)}{\cos(\frac{k\pi}{2n+1})}\cos(\tfrac{2 k\pi}{2n+1}) =: \cJ_{2n+1}(x).
   \]
   We claim again that $\cJ_{2n+1}$ is increasing (see below). However, we now have
   \begin{equation} \label{eq:def:muc}
        \lim_{x \to \infty} \cJ_{2n + 1}(x) = -\frac{1}{2n+1}
	  \sum_{k = 1}^{2n+1}\frac{\cos(\tfrac{2 k\pi}{2n+1})}{\left| \cos(\frac{k\pi}{2n+1}) \right|} =: \mu_c(L).  
   \end{equation}
   If  $\mu \in (0, \mu_c(L))$, we can apply again the intermediate value theorem, and deduce that the equation $\cJ_{2N}(x) = \mu$ has the unique solution $x = \cJ_{2n+1}^{-1}(\mu)$. We deduce as before that there is unique solution of system~\eqref{Eul_Lag__L_fini} in this case. If instead $\mu > \mu_c(L)$, then the system~\eqref{Eul_Lag__L_fini} has no solution.
   
   \medskip
   
   Let us prove that $\cJ_{2n+1}$ is increasing (this will eventually prove that $\mu_c(L) > 0$. Its derivative is given by
   \[
   (2n+1) \cJ_{2n+1}'(x) = - \sum_{k=1}^{2n+1}\frac{\cos\left(\frac{2k\pi}{2n+1}\right)}{\cosh^2\left(x\cos\left(\frac{k\pi}{2n+1}\right)\right)} = - \frac{1}{\cosh^2(x)} - 2 \sum_{k=1}^n \frac{\cos\left(\frac{2k\pi}{2n+1}\right)}{\cosh^2\left(x\cos\left(\frac{k\pi}{2n+1}\right)\right)}.
   \]
   In the last equality, we isolated the $k = 2n+1$ term, and use the change of variable $k' = 2n+1 - k$ for $n+1 \le k \le 2n$. When $1 \le k \le n/2$, we have $\cos\left(\frac{2k\pi}{2n+1}\right) \ge 0$, while $\frac{1}{\sqrt{2}} \le \cos (\tfrac{k \pi}{2n+1}) \le 1$. On the other hand, if $n/2 \le k \le n$, we have $\cos\left(\frac{2k\pi}{2n+1}\right) \le 0$, and $0 \le \cos (\tfrac{k \pi}{2n+1})  \le \frac{1}{\sqrt{2}}$. In both cases, we deduce that
   \[
       \forall k \in \{ 1, \cdots, n \}, \quad  -  \frac{\cos\left(\frac{2k\pi}{2n+1}\right)}{\cosh^2\left(x\cos\left(\frac{k\pi}{2n+1}\right)\right)} \ge - \frac{\cos\left(\frac{2k\pi}{2n+1}\right)}{\cosh^2(\frac{x}{\sqrt{2}})}.
   \]
   Summing over $k$, and using that
   \[
    2 \sum_{k=1}^n \cos\left(\frac{2k\pi}{2n+1}\right) = \sum_{k=1}^{2n} \cos\left(\frac{2k\pi}{2n+1}\right) = \sum_{k=1}^{2n+1} \cos\left(\frac{2k\pi}{2n+1}\right) - 1 = -1,
   \]
   we obtain the lower bound
   \[
    (2n+1) \cJ_{2n+1}'(x) \ge - \frac{1}{\cosh^2(x)} + \frac{1}{\cosh^2(\frac{x}{\sqrt{2}})} \ge 0,
   \]
   which proves that $\cJ_{2n+1}$ is increasing.

   \medskip
   
    Finally, we estimate $\mu_c(L)$, defined in~\eqref{eq:def:muc}. We rewrite $\mu_c(L)$ as
    \[
        \mu_c(L) = \frac{1}{2n+1} \sum_{k=1}^{2n+1} f\left( \tfrac{k}{2n+1} \right) + \frac{1}{2n+1} \sum_{k=1}^{2n+1} \dfrac{1}{\pi | \frac{k}{2n+1} - \frac{1}{2}|}, 
        \quad \text{with} \quad 
        f(s) := \dfrac{\cos(2 \pi s)}{| \cos(\pi s) |} - \frac{1}{\pi | s - \frac{1}{2}|}.
    \]
We recognize a Riemann sum in the first term. Since the function $f$ is integrable on $[0, 1]$ (there is no singularity at $s = \frac{1}{2}$), this term converges to the integral of $f$. For the second term, we recognize a harmonic sum. More specifically, we have
\[
    \frac{1}{2n+1} \sum_{k=1}^{2n+1} \dfrac{1}{\pi | \frac{k}{2n+1} - \frac{1}{2}|} 
    =  \frac{1}{\pi} \sum_{k=1}^{2n+1} \dfrac{1}{ | k - n - \frac{1}{2}|} 
    \sim \frac{2}{\pi} \sum_{k'=1}^n \frac{1}{( k' - \frac{1}{2})} \sim \frac{2}{\pi} \ln(n) \sim \frac{2}{\pi} \ln(L).
\]
  This proves that $\mu_c(L) \sim \frac{2}{\pi}\ln(L)$ at $+\infty$ and completes the proof.
\end{proof}
\bigskip\bigskip

\section{Proofs in the thermodynamic model}

We now focus on the thermodynamic model.

\subsection{Proof of Lemma~\ref{lem:justification_thermo}: Justification of the thermodynamic model}
\label{sec:proof:justification_thermo}

First, we show that this model is indeed the limit of the finite chain model as $L \to \infty$. We denote by $f_\theta^{(2N)}$ the minimum of $g_\theta^{(2N)}$ (so $f_\theta^{(2N)} = \frac{1}{2N} F_\theta^{(2N)}$), and by $\widetilde{f}_\theta$ the minimum of $g_\theta$. Our goal is to prove that $\widetilde{f}_\theta = f_\theta$, where we recall that $f_\theta := \liminf_N f_\theta^{(2N)}$.

\medskip

We denote by $(W_{2N}, \delta_{2N})$ the optimizer of $g_\theta^{(2N)}$, and by $(W_*, \delta_*)$ the one of $g_\theta$. First, from the pointwise convergence $g_\theta^{(2N)}(W, \delta) \to g_\theta(W, \delta)$, we obtain
\[
    \widetilde{f}_\theta = g_\theta (W_*, \delta_*) = \lim_{N \to \infty} g_\theta^{(2N)}(W_*, \delta_*) \ge \lim_{N \to \infty} f_\theta^{(2N)} = f_\theta
\]

For the other inequality, we use that \(h_\theta (x) \le \sqrt{x}+2\theta\ln(2) \), so
\[\begin{aligned} g_\theta ^{(2N)}(W, \delta ) \ge \frac{\mu }{2} \left[ (W-1)^2 + \delta ^2 \right] - \sqrt{W^2 + \delta ^2}-2\theta\ln(2). \end{aligned}\] 
In particular, $g_\theta^{(2N)}$ is lower bounded and coercive, uniformly in $N$. So if $(W_{2N}, \delta_{2N})$ denotes the optimizer of $g_\theta^{(2N)}$, the sequence $(W_{2N}, \delta_{2N})$ is bounded in $\R^2_+$. Up to a not displayed subsequence, we may assume that
\[
    f_\theta = \lim_{N \to \infty} f_\theta^{(2N)} = \lim_{N \to \infty} g_\theta^{(2N)} (W_{2N}, \delta_{2N}), \quad \text{and} \quad
    \lim_{N \to \infty} (W_{2N}, \delta_{2N}) =: (W_\infty, \delta_\infty).
\]
We then have
\[
    f_\theta = \lim_{N \to \infty} g_\theta^{(2N)}(W_{2N}, \delta_{2N}) = \lim_{N \to \infty} g_\theta(W_{2N}, \delta_{2N}) + \lim_{N \to \infty} \left[ g_\theta^{(2N)} - g_\theta \right](W_{2N}, \delta_{2N}).
\]
The first limit converges to $g_\theta(W_\infty, \delta_\infty)$, by continuity of the $g_\theta$ functional. For the second limit, we use that $g_\theta^{(2N)} - g_\theta$ is the difference between an integral and a corresponding Riemann sum. If $\cI_N(s)$ denotes the integrand, this difference is controlled by $\frac{c}{2N} \sup_{s} \| \cI_N'(s) \|$. In our case, $\cI_N(s) = h_\theta(4 W_{2N}^2 \cos^2(\pi s) + 4 \delta_{2N}^2 \sin^2(\pi s))$, whose derivative is uniformly bounded in $N$, since $(W_{2N}, \delta_{2N})$ is bounded. This proves that the last limit goes to zero, hence
\[
    f_\theta = g_\theta(W_\infty, \delta_\infty) \ge \widetilde{f}_\theta.
\]
We conclude that $f_\theta = \widetilde{f}_\theta$. In particular, by uniqueness of the minimizer of $g_\theta$, we must have $(W_\infty, \delta_\infty) = (W_*, \delta_*)$, and the whole sequence $(W_{2N}, \delta_{2N})$ converges to $(W_*, \delta_*)$.


\subsection{Proof of Theorem~\ref{th:main_thermo}: Estimation of the critical temperature.}
\label{sec:proof:main_thermo}

We now study the properties of $\theta_c$, the critical temperature in the thermodynamic limit. Reasoning as in the finite $L$ case, the critical temperature $\theta_c$ can be found by solving the equations in $(W, \theta)$  (compare with~\eqref{Eul_Lag__L_fini})
\begin{equation*}
     \begin{cases} 
         \mu (W - 1) & =\displaystyle \frac{W}{ \pi \theta}  \int_0^{2 \pi}  h' \left(  \dfrac{ W^2 \cos^2 (s)}{\theta^2} \right) \cos^2 (s)  \rd s  \\
         \mu W & = \displaystyle \frac{W}{\pi \theta}  \int_0^{2 \pi}h' \left(  \dfrac{  W^2 \cos^2 (s)}{\theta^2} \right) \sin^2 (s)  \rd s.
     \end{cases}
\end{equation*}
Using again the expression $h'(t) := \frac{\tanh(\sqrt{t})}{\sqrt{t}}$, and splitting the integrals between $(0, 2 \pi)$ into four of size $\pi/2$, this is also
\begin{equation}\label{eq:3eqt_EL}
\begin{cases} 
\mu (W - 1) & = \displaystyle \frac{4}{ \pi }  \int_0^{\pi/2} \tanh \left(  \dfrac{W \cos (s)}{\theta} \right) \cos (s)  \rd s \\
\mu W & = \displaystyle \frac{4}{\pi } \int_0^{\pi/2} \tanh \left( \dfrac{W \cos(s) }{\theta }\right) \frac{\sin^2(s)}{\cos(s)} \rd s.
\end{cases}
\end{equation}

Let us prove that this system always admits a unique solution. The proof is similar to the previous $L \equiv 0 \mod 4$ case. Taking the difference of the two equations gives, with $x := \frac{W}{\theta}$,
	\begin{equation}\label{Eq_of_J}
	\mu = -\frac{4}{\pi}\int_0^{\pi/2} \tanh\left(x\cos(s)\right)\frac{\cos(2s)}{\cos(s)} \rd s =: \cJ\left(x\right),
	\end{equation}
	The function $\cJ$ is derivable on $\R_+$ with derivative  given by
	\[
	\cJ'(x) = \frac{4}{\pi}\int_{0}^{\pi/4}
	\left(\frac{1}{\cosh^2(x\sin(s))} - \frac{1}{\cosh^2(x\cos(s))}\right)\cos(2s) \rd s.
	\] 
	The integrand is positive for all $s \in [0, s/4]$, so $\cJ$ is a strictly increasing function on $\R_+$, and since $\cJ([0 , +\infty)) = [0 , +\infty)$, we get $x = \frac{W}{\theta} = \cJ^{-1}(\mu)$. The first equation of  \eqref{eq:3eqt_EL} gives 
	\[
	\mu(x\theta - 1) = \frac{4}{\pi}\int_{0}^{\pi/2}\tanh\left(x\cos(s)\right)\cos(s)\rd s.
	\]
	This proves that $\theta_c$ is well defined and depends only on $\mu$.\\
	
We now estimate this critical temperature. We are interested in the large $\mu$ limit. First, since $\R\ni u\mapsto\tanh(u)$ is a bounded function, the first equation shows that $\mu(W - 1)$ is uniformly bounded in $\mu$, so $W = 1 + O(\mu^{-1})$ as $\mu \to \infty$. Then, we must have $\theta \to 0$ as $\mu \to \infty$ in order to satisfy the second equation. Using the dominated convergence in the first integral gives
\[
 \frac{4}{ \pi }  \int_0^{\pi/2} \tanh \left(  \dfrac{W}{\theta} \cos (s) \right) \cos (s)  \rd s \xrightarrow[\theta \to 0]{}  \frac{4}{ \pi }  \int_0^{\pi/2} \cos (s)  \rd s = \frac{4}{\pi},
\]
so the first equation gives
\[\begin{aligned} W = 1 + \frac{4}{\pi \mu } + o\left(\frac{1}{\mu}\right). \end{aligned}\] 

We now evaluate the integral of the right-hand side in the second equation, in the limit $\theta \to 0$. It is convenient to make the change of variable $s \mapsto \pi/2 - s$, so we compute
\[
    I(\theta) :=  \int_0^{\pi/2} \tanh \left(  \frac{W}{\theta} \sin(s) \right) \frac{\cos^2(s)}{\sin(s)} \rd s.
\]
In order to evaluate $I(\theta)$ as $\theta \to 0$, we write $I = I_1 + I_2$ with
\[
    I_1 := \int_0^{\pi/2} \tanh \left( \frac{W}{\theta} \sin(s) \right) \frac{\cos(s)}{\sin(s)} \rd s
    \quad \text{and}  \quad
    I_2 := \int_0^{\pi/2} \tanh \left( \frac{W}{\theta} \sin(s) \right) \frac{\cos(s)(\cos(s) - 1)}{\sin(s)} \rd s.
\]

For the first integral, we make the change of variable $u = \frac{W}{\theta} \sin(s)$, and get
\[
    I_1 =  \int_0^{\frac{W}{\theta}} \frac{\tanh \left( u \right)}{u} \rd u = \ln \left( \frac{W}{\theta} \right) + c_1 + o(1), \quad \text{with} \quad c_1 := \int_0^1 \frac{\tanh(u)}{u} + \int_1^\infty \frac{\left(\tanh(u) - 1 \right)}{u} \rd u.
\]
The value of $c_1$ is computed numerically to be $c_1 \approx 0.8188$. For the second integral $I_2$, we remark that the integrand is uniformly bounded in $\theta$ and $s$, so $I_2 = O(1)$. Actually, since $\theta \to 0$, we have, by the dominated convergence theorem that
\[
    I_2 =  \int_0^{\pi/2} \frac{\cos(s)( \cos(s) - 1)}{\sin(s)} \rd s + o(1) = \ln(2) - 1 + o(1).
\]
Altogether, we obtain that
\[
    I(\theta) = \ln \left( \frac{W}{\theta} \right) + c_2 + o(1), \quad \text{with} \quad
    c_2 = c_1 + \ln(2) - 1 \approx 0.512.
\]
Together with the second equation of~\eqref{eq:3eqt_EL}, we obtain
\[
    \mu = \frac{4}{\pi W} \left(  \ln \left( \frac{W}{\theta} \right) + c_2 + o(1) \right)
\]
which gives, as wanted, in the limit $\mu \to \infty$
\[
  \theta_c(\mu) \sim C\exp \left( - \frac{\pi}{4} \mu \right) \mbox{ with } C \approx 0.61385.
\]

\subsection{Proof of Theorem~\ref{th:bifurcation}: study of the phase transition}
\label{bifurcation}

In the previous section, we found the critical temperature. We now study the bifurcation of $\delta$ around this temperature. The critical points of $g_\theta$ are given by the Euler--Lagrange equations
\[
    \begin{cases}
         \mu  \left(  W - 1 \right) & = \displaystyle \frac{W}{ \pi \theta }  \int_0^{2 \pi} h' \left(  \dfrac{ W^2 \cos^2 (s) + \delta^2 \sin^2 (s)}{\theta^2} \right) \cos^2 (s)  \rd s  \\
         \mu W & =  \displaystyle \frac{ W }{\pi \theta}  \int_0^{2 \pi} h' \left(   \dfrac{ W^2 \cos^2 (s) + \delta \sin^2 (s)}{\theta^2} \right) \sin^2 (s)  \rd s.
    \end{cases}
\]
Recall that one can remove the $1$-periodic minimizers by factoring out $\delta$ in the second equation. This gives a set of equation involving $\delta$ through the variable $\Delta := \delta^2$ only. In what follows, we fix $\mu$, and set (we multiply the equations by $\theta/W$ in order to have simpler computations afterwards)
\[
    \cF\left( \theta ; (W, \Delta ) \right) := \begin{cases}
        \displaystyle  \mu \theta \left(  1 - \frac{1}{W} \right)- \frac{1}{ \pi }  \int_0^{2 \pi} h' \left(  \dfrac{ W^2 \cos^2 (s) + \Delta \sin^2 (s)}{\theta^2} \right) \cos^2 (s)  \rd s  \\
        \displaystyle \mu \theta -  \frac{ 1}{\pi}  \int_0^{2 \pi} h' \left(   \dfrac{ W^2 \cos^2 (s) + \Delta \sin^2 (s)}{\theta^2} \right) \sin^2 (s)  \rd s.
    \end{cases}
\]

Recall that $\cF \left( \theta_c ; (W_*, 0) \right) = (0, 0)$, where $W_*$ is the optimal $W$ at the critical temperature. If $\cF \left( \theta ; (W, \Delta) \right) = (0, 0)$ with $\Delta > 0$, the configurations $(W, \pm \sqrt{\Delta})$ are minimizers of $g_\theta$.  If $\cF \left( \theta ; (W, \Delta) \right) = (0, 0)$ with $\Delta < 0$, it does not correspond to a physical solution. \\

We want to apply the implicit function theorem for $\cF$ at the point $(\theta_c ; (W_*, 0))$. In order to do so, we first record all derivatives. We denote by $\cF= (\cF_1, \cF_2)$ the components of $\cF$. The derivatives of $\cF$, evaluated at $\Delta = 0$, $\theta = \theta_c$ and $W = W_*$ are given by
\[
    \begin{cases}
        \partial_W \cF_1 & = \dfrac{\mu \theta_c}{W_*^2} - \dfrac{2 W_*}{\theta_c^2} A \\
        \partial_W \cF_2 & = - \dfrac{2 W_*}{\theta_c^2} B \\
    \end{cases}, \quad
    \begin{cases}
        \partial_\Delta \cF_1 & = - \dfrac{1}{\theta_c^2} B \\
        \partial_\Delta \cF_2 & = - \dfrac{1}{\theta_c^2} C \\
    \end{cases}, \quad \text{and} \quad
    \begin{cases}
        \partial_\theta \cF_1 & = \mu \left(1 - \frac{1}{W_*} \right) + 2 \dfrac{W_*^2}{\theta^3} A \\
       \partial_\theta \cF_2 & = \mu + 2 \dfrac{W_*^2}{\theta_c^3} B \\
    \end{cases}.
\]
where we set (we split the integral in four parts of size $\pi /2$)
\[
    \begin{cases}
        A := \displaystyle \frac{4}{\pi} \int_0^{\pi/2} h'' \left(  \dfrac{ W_*^2 \cos^2 (s)  }{\theta_c^2} \right) \cos^4(s) \rd s  \\
        B := \displaystyle \frac{4}{\pi} \int_0^{\pi/2} h'' \left(   \dfrac{ W_*^2 \cos^2 (s) }{\theta_c^2} \right) \sin^2 (s) \cos^2(s) \rd s \\
        C := \displaystyle \frac{4}{\pi} \int_0^{\pi/2} h'' \left(  \dfrac{ W_*^2 \cos^2 (s)  }{\theta_c^2} \right) \sin^4(s) \rd s.
    \end{cases}
\]
Since $h$ is concave, $A, B$ and $C$ are negative. In addition, by Cauchy-Schwarz, we have 
\begin{equation} \label{eq:BB-AC}
    B^2 \le A C.
\end{equation}

The Jacobian $J := \left( \partial_{(W, \Delta)} \cF \right) (\theta_c; (W_*, 0))$ is of the form
\[
   J = \begin{pmatrix} 
    \frac{\mu \theta_c}{W_*^2}  - \frac{2 W_*}{ \theta_c^2} A & - \frac{1}{\theta_c^2}  B
    \\
    -  \frac{ 2 W_* }{\theta_c^2} B & - \frac{1}{ \theta_c^2}  C
\end{pmatrix}, \quad \text{and} \quad
\det J = - \dfrac{\mu}{W_*^2 \theta_c} C + \dfrac{2W_*}{\theta_c^4} (AC - B^2).
\]
Since $C < 0$ and $B^2 - AC < 0$, we have $\det J > 0$, so $J$ is invertible.
We can therefore apply the implicit function theorem for $\cF$ at $(\theta_c, (W_*, 0))$. There is a function $\theta \mapsto  (W(\theta), \Delta(\theta)) $ so that, locally around $(\theta_c, (W_*, 0))$, we have
\[
    \cF (\theta, (W, \Delta)) = 0 , \quad \text{iff} \quad  (W, \Delta) =  (W(\theta), \Delta(\theta)).
\]
The derivatives $(W'(\theta), \Delta'(\theta))$ are given by
\[
   \begin{pmatrix}
       W'(\theta_c) \\
       \Delta'(\theta_c)  
   \end{pmatrix} = - J^{-1} \begin{pmatrix}
        \partial_\theta \cF_1 \\ \partial_\theta \cF_2
    \end{pmatrix} 
 =   \dfrac{-1}{\det J} \begin{pmatrix} 
      - \frac{1}{ \theta_c^2}  C  &  \frac{1}{\theta_c^2}  B
     \\
       \frac{ 2 W_* }{\theta_c^2} B & \frac{\mu \theta_c}{W_*^2}  - \frac{2 W_*}{ \theta_c^2} A
 \end{pmatrix}
\begin{pmatrix}
     \mu \left(1 - \frac{1}{W_*}\right) + 2 \frac{W_*^2}{\theta_c^3} A \\
    \mu +  2 \frac{W_*^2}{\theta_c^3} B
 \end{pmatrix}.
\]
This gives
\begin{equation} \label{eq:dDelta}
    \Delta'(\theta_c) = \dfrac{-1}{\det J} \left( \frac{2 W_* \mu}{\theta_c^2} \right) \left( (B - A) + \dfrac{\mu \theta_c^3}{2W_*^3}  \right).
\end{equation}
We claim that $B \ge A$ (for the proof see below). This shows that $\Delta'(\theta_c) < 0$. So, restoring the variable $\delta^2$, we have
\[
    \delta^2(\theta) \sim - \Delta'(\theta_c) (\theta_c - \theta)_+ , \quad \text{and finally}, \quad
    \boxed{\delta(\theta) = \sqrt{- \Delta'(\theta_c) } \cdot \sqrt{(\theta_c - \theta)_+} (1 + o(1)).}
\]

It remains to prove that $B \ge A$. This comes from the fact that $h''$ is increasing negative. First, we notice that $| A |$ and $| C |$ are of the form
\[
    | A | = \frac{4}{\pi} \int_0^{\pi/2} f(s) g(s) \rd s, \quad | C | = \frac{4}{\pi} \int_0^{\pi/2} f(s) g(\pi/2 - s),
\]
with $f(s) := \left| h''(W_*^2 \cos^2(s)/\theta_c^2) \right| $ and $g(s) := \cos^4(s)$. The functions $f$ and $g$ are both decreasing on $[0, \frac{\pi}{2}]$. By re-arrangement, we deduce that $| A | > | C |$. Actually, we have
\[
    |A | - | C | = \frac{4}{\pi} \int_0^{\pi/4} \left( f(s)  - f\left(\frac{\pi}{2} - s\right) \right)
    \left( g(s)  - g\left(\frac{\pi}{2} - s\right) \right) > 0.
\]
Together with Cauchy-Schwarz in~\eqref{eq:BB-AC}, this gives $| B |^2 \le | A | \cdot | C | < | A |^2$, since $A$ and $B$ are negative, we get $B > A$, as wanted. This concludes the proof of Theorem~\ref{th:bifurcation}.


\appendix
\section{Gain of energy in the thermodynamic limit}\label{gain_of_energy}

In this section, we prove that the gain of energy due to Peierls dimerization is exponentially small in $\mu$. We focus on the thermodynamic limit case (although the proof is simililar in the $L \in 2 \N$ case). We also focus only on the null temperature case $\theta = 0$. In this case, the thermodynamic energy reads
\begin{equation}\label{CH1_42}
    g_0(W ,\delta) = \frac{\mu}{2} ((W - 1)^2 + \delta^2) - \frac{4}{\pi}\int_{0}^{\pi/2}\sqrt{W^2\sin^2{(s)} + \delta^2\cos^2{(s)}} \rd s.
\end{equation}
We introduce
\[
    f_0 := \min \left\{ g_0(W, \delta), \ W \ge 0, \ \delta \ge 0\right\}, \quad \text{and} \quad    f_{0, {\rm per}} := \min \left\{ g_0(W, 0), \ W \ge 0 \right\}.
\]
In other words, $f_0$ is the minimum of $g_0$ over $2$--periodic (and all) configurations, and $f_{0, {\rm per}}$ is the minimum over $1$-periodic configurations. We prove the following

\begin{theorem}
    There is $C > 0$ such that, for all $\mu$ large enough,
        \[
        0 < f_{0, {\rm per}} - f_0 \le C \re^{- \frac{\pi}{2} \mu}.
        \]
\end{theorem}

In other words, the energy gained by the Peierls distorsion is exponentially small in the $\mu$ parameter. The first inequality states that in the thermodynamic limit at null temperature, the minimizers are always dimerized, as first proved by Kennedy and Lieb~\cite{kennedy2004proof}.

\begin{proof}
    Let us first compute $W_1$, the optimizer of $g_0(W, 0)$. This is simply the minimum of 
    $$
    g_0(W,0) = \frac{\mu}{2} (W - 1)^2 - \frac{4}{\pi}\int_0^{\pi/2}\sqrt{W^2\sin^2(s)} \rd s = \frac{\mu}{2}\mu(W - 1)^2 - \frac{4}{\pi}W.
    $$
    The minimizer  satisfies $\mu(W_1 - 1) = \frac{4}{\pi}$, hence $W_1 =  1 + \frac{4}{\pi \mu}$. In particular,
    $$  f_{0, {\rm per}} = -\frac{4}{\pi} - \frac{8}{\pi^2\mu}.
    $$
    We now compute the energy gain from the breaking of periodicity. For $(W, \delta)$ a trial pair, we write $W = W_1 + \varepsilon.$ 
    We assume that $g_0(W, \delta) < g_0(W_1, 0)$. Then
    \begin{align*}
         g_0(W, \delta) - g_0(W_1, 0)  &= \frac{\mu}{2}(\varepsilon^2 + \delta^2) \\ &\quad -\frac{4W_1}{\pi}\int_0^{\pi/2}\left[ \sqrt{ \frac{(W_1 + \varepsilon)^2}{W_1^2} + \frac{\delta^2}{W_1^2}\cot^2(s)} - 1 - \frac{\varepsilon}{W_1} \right]\sin(s) \rd s\\
        &\geq \frac{\mu}{2}(\varepsilon^2 + \delta^2) -  \frac{4\delta}{\pi},\;\hbox{ so }\;\delta<\frac{8}{\pi\mu}\;\hbox{ and }\;|\varepsilon|<\frac{4}{\pi\mu}.
    \end{align*}
    To compute the integral, we make the change of variable $u=\cos(s),$ and get that the integral equals
    \begin{align*}
        \frac{W_1 + \varepsilon}{W_1}\left( \int_0^1\sqrt{1+ \frac{au^2}{1-u^2}} \rd u -1 \right), \mbox{ with } a:= \left(\frac{\delta}{W_1 + \varepsilon}\right)^2.
    \end{align*}
    Using that
    \[
        \int_0^1\sqrt{1+ \frac{au^2}{1-u^2}} \rd u  = E(1 - a) = 1 + \left(\frac{-\ln(a)}{4}  -\frac{1}{4} + \ln(2) \right)a + O(a^2),
    \]
    where $E$ is a complete elliptic integral of the second kind, we get
     \begin{align*}
        g_0(W, \delta) - g_0(W_1, 0) &=   \frac{\mu}{2}(\varepsilon^2 + \delta^2)  - \frac{4\delta^2}{\pi(W_1 + \varepsilon)} \left[-\frac{1}{2}\ln\left(\frac{\delta}{W_1 + \varepsilon}\right) -\frac{1}{4} + \ln(2) +  O(a)\right]\\
         & =  \frac{1}{2}\mu(\varepsilon^2 + \delta^2) - \frac{2}{\pi W_1 }\delta^2\ln(\delta^{-1})\left(1+o(1)\right).
    \end{align*}
    We now minimize the right-hand side. For large $\mu$, we have $W_1=1+o(1)$ and the minimization in $\varepsilon$ gives $\varepsilon = 0$. So 
    $$g_0(W, \delta) - g_0(W_1, 0) \geq \delta^2\left(\frac{\mu}{2} - \frac{2\ln(\delta^{-1})}{\pi}\left(1+o(1)\right)\right).$$
    We optimize the right-hand side by taking $\delta = \re^{-(\frac{\pi}{4}\mu + \frac{1}{2})} $, and this completes the proof.
    
\end{proof}

\bibliographystyle{plain}
\bibliography{biblio}

\end{document}